\documentclass[letterpaper,journal]{IEEEtran}

\usepackage{amsmath,amsfonts,amssymb}
\usepackage{array}
\usepackage{graphicx}
\usepackage[caption=false,font=normalsize,labelfont=sf,textfont=sf]{subfig}
\usepackage{textcomp}
\usepackage{stfloats}
\usepackage{url}
\usepackage{verbatim}
\usepackage{cite}
\usepackage{tikz}
\usepackage{balance}
\usepackage{booktabs}

\usetikzlibrary{
arrows.meta,
calc,
positioning,
fit,
angles,
quotes,
decorations.pathmorphing
}

\newcommand{\InsertCampaignFigure}[3]{%
\includegraphics[width=#2]{figures/#1}%
}

\newtheorem{lemma}{Lemma}

\begin{document}

% ==========================================================
% TITLE
% ==========================================================

\title{LiDAR-Derived Surface Priors for Multimodal Sensing-Assisted
NLoS Beam Search in Indoor 60-GHz Networks}

\author{
Amod Ashtekar,
Dalton Davis,
Rafaela Lomboy,
Omar Ibrahim,
Raj Sai Sohel Bandari,
and Mohammed E. Eltayeb,
~\IEEEmembership{Senior Member,~IEEE,}%
\thanks{
Amod Ashtekar, Rafaela Lomboy, Omar Ibrahim, Raj Sai Sohel Bandari,
and Mohammed E. Eltayeb are with the Department of Electrical and
Electronic Engineering, California State University, Sacramento,
Sacramento, CA, USA.
}%
\thanks{
Dalton Davis is with the Department of Computer Science,
University of California, Davis, Davis, CA, USA.
}%
\thanks{
Corresponding author: Mohammed E. Eltayeb
(e-mail: mohammed.eltayeb@csus.edu).
}%
\thanks{
This material is based upon work supported by the National
Science Foundation under Grant NSF-2243089.
}%
}

\maketitle

% ==========================================================
% ABSTRACT
% ==========================================================

\begin{abstract}
Highly directional 60-GHz Internet-of-Things (IoT) links can exploit
naturally occurring indoor surfaces to sustain connectivity under blockage.
Identifying viable non-line-of-sight (NLoS) paths, however, can require
extensive RF beam training. This paper investigates whether LiDAR can reduce this
search overhead by providing a surface-aware prior without assuming a direct
mapping between optical return and mmWave reflection. The proposed framework
uses LiDAR-derived geometry and return statistics to rank candidate
propagation directions, while RF measurements remain responsible for final
beam selection. The experimental validation is organized in three stages to
separate descriptor robustness, cross-modal association, and beam-search
performance. Controlled LiDAR measurements first quantify how geometric and
radiometric surface descriptors vary with acquisition geometry. Matched
LiDAR and 60-GHz measurements in an L-shaped corridor then determine whether
these descriptors are associated with the measured surface-mediated RF
response under a prescribed NLoS interaction. Finally, a separate room-scale
campaign evaluates the resulting prior using exhaustive TX--RX beam maps
without prescribing the underlying propagation mechanism. The measurements
show that descriptor reliability depends on acquisition geometry and
point-cloud representation, and that LiDAR and RF surface responses exhibit
cross-modal association without supporting deterministic RF-power
prediction.  In the room experiment, local three-ring 3-D planarity retains a beam within
3~dB of exhaustive search at \(74.5\%\) of the measured locations while
reducing RF beam-pair probing by \(72\%\) relative to exhaustive probing over
the candidate search region.  These results establish LiDAR-derived
local surface structure as a communication-oriented prior for concentrating
RF probing and reducing mmWave beam-search uncertainty.
\end{abstract}


% ==========================================================
% I. INTRODUCTION
% ==========================================================

\section{Introduction}
\label{sec:introduction}

\IEEEPARstart{M}{illimeter}-wave (mmWave) communication enables high-rate
connectivity for emerging Internet-of-Things (IoT) applications, including
industrial automation, mobile robots, smart infrastructure, and dense indoor
networks \cite{sahoo2019mmwaveiot,dinh2022device,solomitckii2018factory}.
At 60~GHz, however, reliable connectivity requires highly directional
transmission that is vulnerable to blockage and beam misalignment
\cite{cano2021industrial,kim2024widebeam}. In sensing-enabled IoT
deployments, such as industrial gateways and robotic platforms, blockage
therefore creates both a coverage problem and a beam-reacquisition problem:
an alternative propagation direction must be identified while limiting the
RF resources consumed by beam training. 
Indoor environments contain walls, boards, cabinets, glass, metallic objects,
and other naturally occurring surfaces that can support alternative
non-line-of-sight (NLoS) propagation. Measurements of building materials and
indoor environments have demonstrated substantial reflection and scattering
at millimeter-wave frequencies
\cite{langen1994reflection,lu2014propagation,tip2019indoor,xing2019indoor}.
Moreover, passive reflectors have been shown to extend mmWave coverage around
corners and into blocked regions \cite{khawaja2020passive}. The networking
challenge is therefore how to identify and prioritize promising NLoS
directions without repeatedly probing the complete TX--RX beam codebook.

RF-only exhaustive, hierarchical, and sparse beam-training methods address
this problem through additional channel measurements
\cite{xue2024beammanagement,alkhateeb2014channel,
eltayeb2015opportunistic}. Repeated beam acquisition consumes radio resources
that would otherwise support data transmission and link maintenance,
motivating the use of side information to reduce the directional search space
\cite{xiao2025localization}. Multimodal environmental sensing provides a
complementary approach where sensing-enabled IoT infrastructure can first reduce
spatial uncertainty and concentrate RF probing on a smaller set of candidate
directions. This architecture is particularly relevant to IoT deployments in
which gateways, access points, robots, or instrumented infrastructure already
observe the surrounding environment while supporting directional wireless
links. It also relates to broader integrated sensing and communication (ISAC)
and multimodal sensing-assisted communication, where heterogeneous
environmental information can support beamforming and communication-resource
allocation \cite{liu2022isac,patel2024multimodal,peng2026simac}.

LiDAR is attractive for this role because it directly measures range and
angle, resolves finite surface extent and local orientation, and provides an
instantaneous observation of the surrounding environment.  In this paper,  we consider a
sensing-enabled IoT gateway, access point, or robotic endpoint equipped with
a directional 60-GHz radio and co-located LiDAR.  Unlike map-assisted approaches, the proposed method does not require a
pre-existing geometric map or prior knowledge of the propagation surface at
the current location. Candidate surface directions are identified
opportunistically from the current LiDAR observation.  After LiDAR-to-RF angular registration,
the observed surfaces are associated with candidate RF directions and ranked
before RF probing. Fusion therefore occurs at the decision level: LiDAR
reduces the directional search space by prioritizing candidate directions
based on the observed environment, while measured mmWave power determines
the final communication beam direction. The objective is not to reconstruct the
propagation environment or predict the RF channel from LiDAR, but to use
instantaneous environmental sensing to reduce beam-search uncertainty and RF
training expenditure.

Prior work has used LiDAR and other sensing modalities for learned
sensor-to-beam prediction and multimodal beam selection
\cite{klautau2019lidar,mashhadi2021federated,jiang2023future,
patel2024multimodal,tariq2024quantum,deng2026reliable}. Model-based
approaches have also exploited explicit scene structure. SpaceBeam, for
example, reconstructs a three-dimensional environment and uses ray-based
propagation modeling for beam management \cite{woodford2021spacebeam}, while
point-cloud observations combined with RF measurements have been used to
estimate material permittivity at 60~GHz \cite{virk2018permittivity}.
These approaches demonstrate the value of environment awareness, but they
leave a central modeling question: \emph{what LiDAR-observed surface information can be used to opportunistically
identify and prioritize candidate mmWave propagation directions?}

This distinction is fundamental. A LiDAR generally measures a monostatic or
near-monostatic optical return, whereas a useful mmWave NLoS path involves a
bistatic TX--surface--RX interaction. LiDAR intensity depends on sensing
range, incidence angle, optical response, footprint, detection threshold, and
sensor processing \cite{hofle2007intensity,tan2016tls}. The 60-GHz
surface-mediated field additionally depends on RF permittivity and loss,
polarization, wavelength-normalized surface variation, incidence and
departure directions, and the scattering mechanism
\cite{sato1997materials,goulianos2017scattering,
degliesposti2007scattering}. Consequently, a strong LiDAR return cannot in
general be interpreted as a strong 60-GHz reflection. The objective of this work is therefore not to infer material identity or
predict RF received power from LiDAR, but to determine whether
LiDAR-observed geometry and surface-return structure can rank naturally
occurring surfaces as candidates for RF probing.

Our prior work explored LiDAR-aided NLoS path identification and 
backscatter-guided alignment
\cite{davis2025nlos,ashtekar2026backscatter}. The present
work extends these studies through a unified surface-prior formulation,
controlled cross-view and cross-modal measurements, and an independent
room-scale evaluation using exhaustive TX--RX beam maps. The paper addresses
three research questions:

\begin{itemize}

\item[\textbf{RQ1}]
Which geometric and radiometric LiDAR descriptors retain surface-dependent
information across changes in viewing geometry and sensing conditions?

\item[\textbf{RQ2}]
Do oblique 3-D LiDAR surface descriptors contain information associated with
the measured 60-GHz surface-mediated response of the same surfaces under a
controlled NLoS geometry?

\item[\textbf{RQ3}]
Can instantaneous local 3-D LiDAR observations reduce RF beam-search
uncertainty without a pre-existing environmental map or prior knowledge of
the relevant propagation surface?

\end{itemize}
Together, the three research questions progress from LiDAR descriptor
robustness, to cross-modal surface association, and finally to
communication-level beam-search performance. The main contributions are summarized as follows:

\begin{itemize}

\item We formulate LiDAR-assisted NLoS beam search as a
\emph{surface-informed prior} that opportunistically ranks candidate RF
directions using locally observed indoor surface structure for targeted RF
probing.

\item We relate physical surface-height variation to mmWave phase dispersion
and distinguish RF surface roughness from LiDAR-observed geometric
dispersion, coherence, and planarity.

\item We characterize the acquisition dependence of geometric and
radiometric LiDAR descriptors and examine their association with the
measured 60-GHz response of matched indoor surfaces.

\item We evaluate LiDAR-guided beam selection in a separate indoor
environment and compare it with alternative LiDAR representations,
geometry-informed selection, and equal-budget random probing.

\end{itemize}

The remainder of the paper is organized as follows.
Section~\ref{sec:proposed} develops the LiDAR-derived surface-prior framework
and its physical motivation. Section~\ref{sec:experimental_methodology}
describes the measurement campaigns and evaluation protocol.
Section~\ref{sec:controlled_validation} addresses RQ1 and RQ2, and
Section~\ref{sec:room_validation} addresses RQ3.
Section~\ref{sec:discussion} discusses the resulting design implications,
and Section~\ref{sec:conclusion} concludes the paper.

% ==========================================================
% II. LIDAR-DERIVED SURFACE PRIORS FOR MMWAVE BEAM SEARCH
% ==========================================================

\section{LiDAR-Derived Surface Priors for mmWave Beam Search}
\label{sec:proposed}

The proposed framework uses complementary LiDAR and mmWave observations to
exploit naturally occurring indoor surfaces for NLoS communication.  The LiDAR observation provides three complementary forms of information:
surface geometry, apparent optical return, and local geometric/radiometric
structure. These observations depend on sensing location and viewing
geometry. Fig.~\ref{fig:room_lidar_intensity_location_dependence} illustrates
this dependence using representative returned-intensity maps from three
receiver locations in an indoor environment.  Although the physical environment is unchanged, the
measured LiDAR response varies with range, viewing angle, visible support,
footprint, and occlusion.  The communication relevance of these descriptors is evaluated experimentally
in Sections~IV and~V.

\begin{figure}[t]
\centering
\subfloat[RX4]{%
    \includegraphics[width=0.32\columnwidth]
    {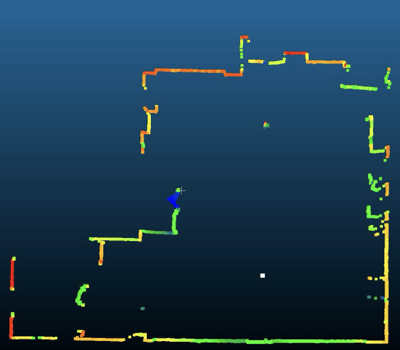}%
    \label{fig:lidar_intensity_rx04}}
\hfill
\subfloat[RX17]{%
    \includegraphics[width=0.32\columnwidth]
    {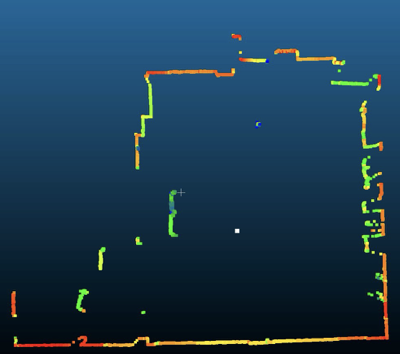}%
    \label{fig:lidar_intensity_rx17}}
\hfill
\subfloat[RX49]{%
    \includegraphics[width=0.32\columnwidth]
    {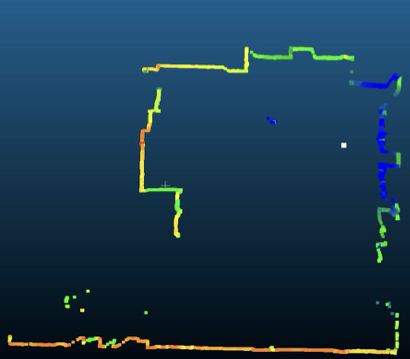}%
    \label{fig:lidar_intensity_rx49}}
\caption{Representative LiDAR returned-intensity maps at three receiver
locations in an indoor environment. Colors indicate relative returned
optical power within each panel. } %The white square marker indicates the LiDAR location.
\label{fig:room_lidar_intensity_location_dependence}
\end{figure}

% ----------------------------------------------------------
\subsection{Physical Motivation for Surface-Assisted Communication}
\label{subsec:physical_motivation}
% ----------------------------------------------------------

The use of LiDAR-observed surfaces as an RF prior is motivated by the
dependence of surface-mediated propagation on the physical structure of the
reflecting surface. For a surface with physical root-mean-square (RMS) height
variation $\sigma_h$, the classical Rayleigh/Beckmann--Kirchhoff rough-surface
approximation gives the coherent specular field roughness factor
$\rho_s=\exp[-8(\pi\sigma_h\cos\vartheta_{\mathrm{inc}}/
\lambda_{\mathrm{RF}})^2]$~\cite{Ju2019,Piesiewicz2007}.
Since reflected power is proportional to the squared magnitude of the
field coefficient, the corresponding coherent reflected-power factor is
\begin{equation}
\eta_{\mathrm{coh}}
=
|\rho_s|^2
=
\exp\!\left[
-\left(
\frac{4\pi\sigma_h\cos\vartheta_{\mathrm{inc}}}
{\lambda_{\mathrm{RF}}}
\right)^2
\right],
\label{eq:specular_coherent_factor}
\end{equation}
where $\eta_{\mathrm{coh}}$ is the coherent reflected-power factor,
$\lambda_{\mathrm{RF}}$ is the RF wavelength, and
$\vartheta_{\mathrm{inc}}$ is the incidence angle measured from the
local surface normal. Equation~\eqref{eq:specular_coherent_factor} shows that
physical surface-height variation becomes increasingly important as its
magnitude grows relative to the wavelength. The associated phase mechanism
is summarized by the following lemma.

\begin{lemma}
\label{lem:rms_phase_dispersion}
Consider a locally planar surface with zero-mean height deviations \(h_m\),
\(m=1,\ldots,M\), and RMS height
\(\sigma_h=(M^{-1}\sum_{m=1}^{M}h_m^2)^{1/2}\).
For a wave of wavelength \(\lambda\), the RMS phase dispersion caused by
the surface-height variation is
\begin{equation}
\sigma_{\phi}(\lambda)
=
\frac{2\pi\sigma_h}{\lambda}
\left|
\cos\theta_i+\cos\theta_s
\right|,
\label{eq:rms_phase_dispersion}
\end{equation}
where \(\theta_i\) and \(\theta_s\) are the incident and scattered angles,
respectively, measured from the local surface normal. Thus, phase dispersion
increases with the wavelength-normalized height variation
\(\sigma_h/\lambda\).
\end{lemma}

\begin{IEEEproof}
Let $k=\frac{2\pi}{\lambda}$, denote the wavenumber. For a small local height deviation $h_m$ relative
to the propagation distances, the incident and scattered path lengths
change to first order by $h_m\cos\theta_i$ and $h_m\cos\theta_s$,
respectively, up to an overall sign determined by the direction of the
surface displacement. The resulting total path-length perturbation is
therefore
\[
\Delta \ell_m
=
h_m\left(\cos\theta_i+\cos\theta_s\right).
\]
The corresponding phase perturbation is
\[
\phi_m
=
k\Delta \ell_m
=
\frac{2\pi}{\lambda}
h_m\left(\cos\theta_i+\cos\theta_s\right).
\]
Since the height deviations are zero mean, the phase perturbations are
also zero mean and their RMS value becomes
\begin{align}
\sigma_{\phi}
&=
\left(
\frac{1}{M}\sum_{m=1}^{M}\phi_m^2
\right)^{1/2}
\\
&=
\left[
\frac{1}{M}\sum_{m=1}^{M}
\left(
\frac{2\pi}{\lambda}
h_m(\cos\theta_i+\cos\theta_s)
\right)^2
\right]^{1/2}
\\
&=
\frac{2\pi}{\lambda}
\left|\cos\theta_i+\cos\theta_s\right|
\left(
\frac{1}{M}\sum_{m=1}^{M}h_m^2
\right)^{1/2}.
\end{align}
Using the definition
\[
\sigma_h
=
\left(
\frac{1}{M}\sum_{m=1}^{M}h_m^2
\right)^{1/2}
\]
gives
\[
\sigma_{\phi}(\lambda)
=
\frac{2\pi\sigma_h}{\lambda}
\left|\cos\theta_i+\cos\theta_s\right|,
\]
which proves~\eqref{eq:rms_phase_dispersion}.

For specular reflection, $\theta_s=\theta_i$, and hence
\begin{equation}
\sigma_{\phi}(\lambda)
=
\frac{4\pi\sigma_h\cos\theta_i}{\lambda}.
\label{eq:specular_phase_rms}
\end{equation}
Thus, larger wavelength-normalized height variation produces larger
phase differences among contributions from different parts of the surface
and can reduce their coherent concentration in the specular direction.
\end{IEEEproof}
Lemma~\ref{lem:rms_phase_dispersion} highlights the role of wavelength.
If a surface is physically smooth relative to the much shorter LiDAR
wavelength, the same physical height variations are necessarily small
relative to the 60-GHz wavelength. However, a practical LiDAR point cloud
does not directly measure the microscopic $\sigma_h$ in
\eqref{eq:specular_coherent_factor}. Its observed geometry also depends
on ranging resolution, beam footprint, sampling density, viewing geometry,
edge mixing, and sensor processing. Consequently, high LiDAR planarity or
coherence does not by itself establish physical RF smoothness.

Likewise, LiDAR observations with greater geometric variation are not
discarded a priori, since structured or irregular surfaces may redistribute
RF energy into useful off-specular directions. We therefore treat the
LiDAR-derived quantities below as observation-domain surface cues and
determine their communication relevance experimentally rather than using
them as direct estimates of RF roughness or reflection loss.

% ----------------------------------------------------------
\subsection{LiDAR Surface Characterization for Communication}
\label{subsec:lidar_surface_descriptors}
% ----------------------------------------------------------

Let \(S\) denote a local LiDAR observation containing \(N_S\) returns, with
\(m\) indexing a return in \(S\). Each return provides the 3-D coordinate
\(\mathbf p_m=[x_m,y_m,z_m]^{\mathsf T}\), range \(d_{L,m}\), and
linear returned intensity \(I_{L,m}\).

A nominally range-normalized apparent-return quantity is defined as $w_{L,m}
=
I_{L,m}d_{L,m}^{2}.$
%
%\begin{equation}
%w_{L,m}
%=
%I_{L,m}d_{L,m}^{2}.
%\label{eq:range_normalized_definition}
%\end{equation}
This compensation removes only the nominal inverse-square range dependence.
The quantity can still depend on viewing angle, footprint, local surface
orientation, and sensor response and is therefore treated as an
apparent-return cue rather than an intrinsic material property.

\subsubsection{Geometric Return Dispersion}

To characterize how closely the measured LiDAR points follow a local plane,
let \(\delta h_m\) denote the perpendicular residual of return \(m\) from
the best-fit plane through \(S\). The geometric return dispersion is

\begin{equation}
R_{\mathrm{rms},L}(S)
=
\sqrt{
\frac{1}{N_S}
\sum_{m\in S}
\delta h_m^2
}.
\label{eq:lidar_return_dispersion}
\end{equation}
A smaller \(R_{\mathrm{rms},L}(S)\) indicates that the measured returns lie
more closely around a common local plane. This sensor-domain quantity is
distinct from the physical RMS height \(\sigma_h\) used in
Section~\ref{subsec:physical_motivation}.

\subsubsection{Local Surface Organization}

Local surface organization is characterized by how the LiDAR points spread
in space. For a horizontal LiDAR slice, let
\(\lambda_1\geq\lambda_2\geq0\) denote the point-cloud spread along its two
principal directions. A well-organized surface slice has a dominant spread
along the surface and relatively little spread transverse to it. The 2-D
coherence is defined as

\begin{equation}
\Pi(S)
=
\frac{\lambda_1-\lambda_2}
{\lambda_1+\lambda_2+\epsilon_\Pi},
\qquad
0\leq\Pi(S)\leq1,
\label{eq:lidar_surface_coherence}
\end{equation}
where \(\epsilon_\Pi>0\) prevents division by zero. A larger \(\Pi(S)\)
indicates a more clearly organized horizontal surface slice.

For a local 3-D point cloud, let
\(\lambda_1\geq\lambda_2\geq\lambda_3\geq0\) denote the spread along its
three principal directions. For a planar surface, the two largest spreads
lie primarily along the surface, whereas the smallest spread is associated
with the surface-normal direction. The 3-D planarity is defined as

\begin{equation}
\Pi_{3\mathrm D}(S)
=
\frac{\lambda_2-\lambda_3}
{\lambda_1+\epsilon_\Pi}.
\label{eq:lidar_3d_planarity}
\end{equation}
A larger \(\Pi_{3\mathrm D}(S)\) indicates that the measured returns more
closely form a local planar surface.

\subsubsection{Radiometric Return Structure}

Let the mean linear returned intensity within \(S\) be

\begin{equation}
\bar I_L(S)
=
\frac{1}{N_S}
\sum_{m\in S}I_{L,m}.
\label{eq:lidar_mean_intensity}
\end{equation}
The returned-intensity coefficient of variation is

\begin{equation}
\mathrm{CV}_{I}(S)
=
\frac{
\sqrt{
N_S^{-1}
\sum_{m\in S}
\left(I_{L,m}-\bar I_L(S)\right)^2
}
}{
\bar I_L(S)+\epsilon_I
},
\label{eq:lidar_intensity_cv}
\end{equation}
where \(\epsilon_I>0\) prevents numerical instability. A larger
\(\mathrm{CV}_{I}(S)\) indicates greater variation in the measured optical
return across \(S\).

Localized strong-return behavior is summarized by the linear max-to-mean
ratio

\begin{equation}
S_{\max/\mu}(S)
=
\frac{\max_{m\in S}I_{L,m}}
{\bar I_L(S)}.
\label{eq:lidar_max_mean_response}
\end{equation}

A larger $S_{\max/\mu}(S)$ indicates that one or more localized returns are
strong relative to the mean return. For compact reporting, the values in the
tables and figures are expressed in decibels as
$10\log_{10} S_{\max/\mu}(S)$.  Collectively, \(R_{\mathrm{rms},L}\), \(\Pi\), and
\(\Pi_{3\mathrm D}\) describe the measured geometric structure, while
\(\mathrm{CV}_{I}\) and \(S_{\max/\mu}\) describe radiometric structure.
The quantity \(w_{L,m}\) provides a separate range-compensated
apparent-return cue. These descriptors are evaluated individually as
surface-ranking cues rather than as direct predictors of RF received power.

% ----------------------------------------------------------
\subsection{Beam-Level Mapping and RF Verification}
\label{subsec:performance_measures}
% ----------------------------------------------------------

The LiDAR point cloud is more densely sampled than the RF beam codebook.
Accordingly, the descriptors above are evaluated over LiDAR returns
associated with each candidate RF receive direction.   Let $\phi_{L,m}$ denote
the measured azimuth of LiDAR return $m$, and let
$\mathcal{A}(\cdot)$ denote the angular-registration mapping from the LiDAR
coordinate frame to the RF coordinate frame. The registered LiDAR azimuth is
$
\widetilde{\phi}_{L,m}
=
\mathcal{A}\!\left(\phi_{L,m}\right).
$
Each return is then associated with the RF receive direction whose boresight
angle is nearest to the registered LiDAR azimuth,
\begin{equation}
r_m
=
\arg\min_{r\in\mathcal R}
d_{\mathrm{ang}}
\!\left(
\widetilde{\phi}_{L,m},\phi_r
\right),
\label{eq:lidar_rf_mapping}
\end{equation}
where $\phi_r$ is the boresight angle of receive direction $r$ and
$d_{\mathrm{ang}}(\cdot,\cdot)$ denotes wrapped angular distance.

Let \(\mathcal T\) and \(\mathcal R\) denote the evaluated TX and RX beam
codebooks, respectively. At measurement location \(i\), let
\(P_{i,t,r}\) denote the measured linear RF power for beam pair
\((t,r)\), where \(t\in\mathcal T\) and \(r\in\mathcal R\).
The exhaustive-search reference is

\begin{equation}
P_i^\star
=
\max_{t\in\mathcal T,\;r\in\mathcal R}
P_{i,t,r}.
\label{eq:rf_reference}
\end{equation}

For TX-independent LiDAR selection, let
\(\mathcal S_i^R(K_R)\subseteq\mathcal R\) contain the \(K_R\)
highest-ranked RX directions. All TX beams are then swept for the retained
RX directions, giving

\begin{equation}
P_{i,\mathrm{sel}}
=
\max_{\substack{
t\in\mathcal T\\
r\in\mathcal S_i^R(K_R)
}}
P_{i,t,r},
\label{eq:selected_rf_power}
\end{equation}
where \(K_R\) is the number of retained RX directions.  The beamforming loss relative to exhaustive search is $L_i
=
10 \log_{10}\left(
\frac{P_i^\star}{P_{i,\mathrm{sel}}}
\right)$ dB.

\section{Experimental Methodology}
\label{sec:experimental_methodology}

The experimental setup spans three campaigns of increasing complexity.
The first establishes a near-normal LiDAR reference over representative
indoor surfaces. The second uses matched full-3-D LiDAR observations of five
surfaces under an approximately \(45^\circ\) geometry and relates them to a
controlled 60-GHz NLoS field in an L-shaped corridor. The third evaluates
LiDAR-guided beam selection in a separate room-scale environment using
exhaustive TX--RX beam measurements and full 3-D LiDAR scans. The
near-normal and oblique LiDAR measurements jointly address RQ1, the matched
LiDAR--mmWave corridor measurements address RQ2, and the room-scale
campaign addresses RQ3.

\subsection{Near-Normal LiDAR Surface Campaign}
\label{subsec:normal_campaign}

A Quanergy M8 LiDAR is used to characterize 16 indoor surfaces under an
approximately near-normal observation geometry. Representative drywall and
copper measurement setups are shown in Fig.~\ref{fig:lidar_setup}.  The near-normal reference is collected at approximately 3~ft~4~in distance. The full
3-D surface point cloud is used to compute the plane-fit RMS and the
radiometric descriptors. This campaign establishes controlled
surface-dependent observations; it is not an RF reflector test or a semantic
material-classification experiment.

\begin{figure}[t]
\centering
\begin{minipage}[t]{0.48\columnwidth}
\centering
\InsertCampaignFigure{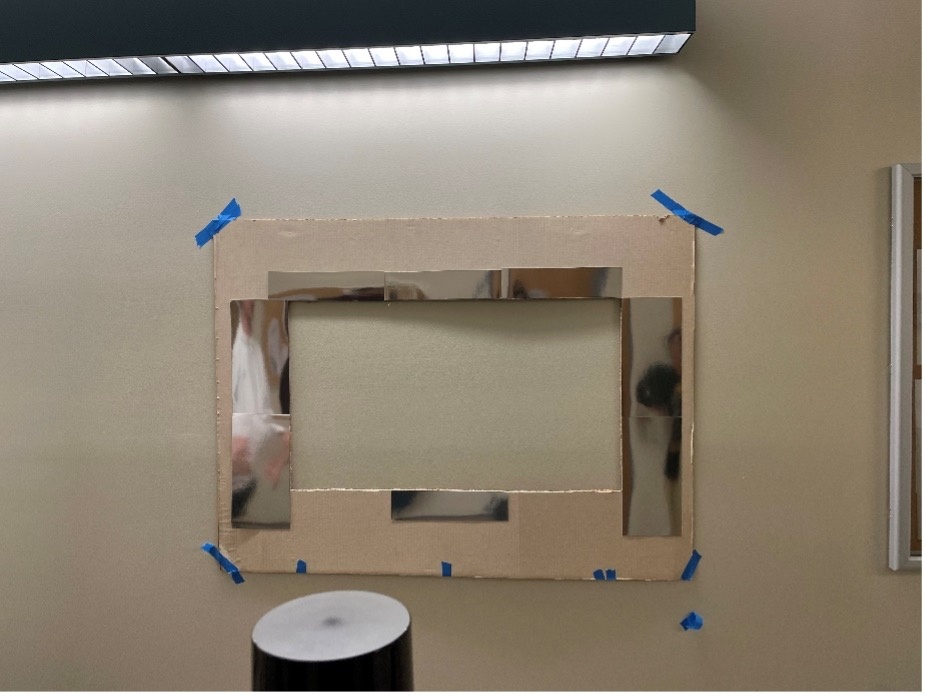}{\linewidth}
{Drywall LiDAR material-characterization setup.}
\par\vspace{2pt}\small (a) Drywall.
\end{minipage}
\hfill
\begin{minipage}[t]{0.48\columnwidth}
\centering
\InsertCampaignFigure{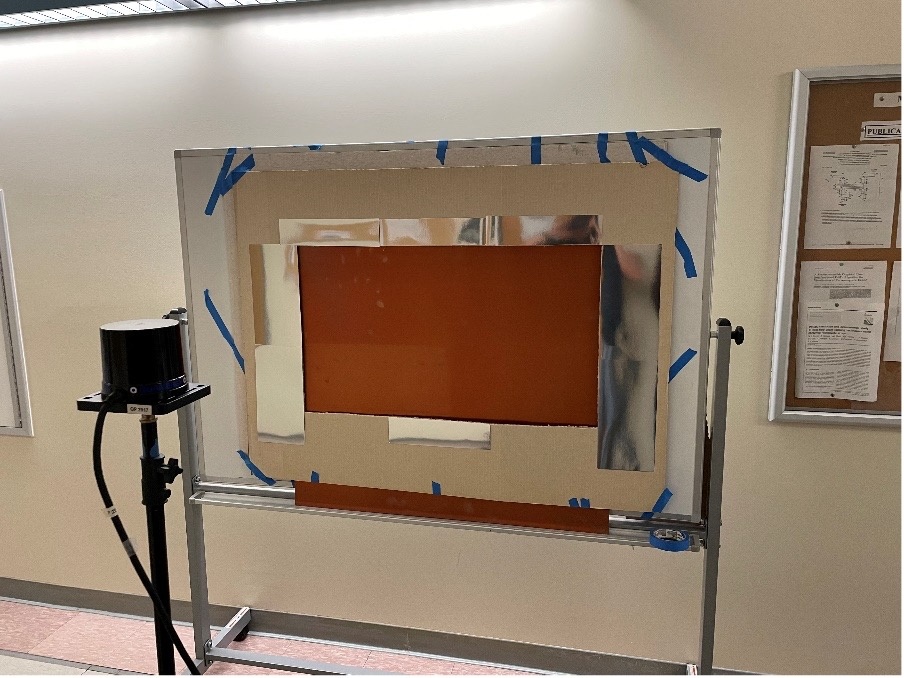}{\linewidth}
{Copper LiDAR material-characterization setup.}
\par\vspace{2pt}\small (b) Copper.
\end{minipage}
\caption{Controlled near-normal LiDAR surface-characterization setup.}
\label{fig:lidar_setup}
\end{figure}

\subsection{Oblique 3-D LiDAR/mmWave Corridor Campaign}
\label{subsec:angled_campaign}

Five matched surfaces (copper, silver, rough wood, cardboard, and smooth
compressed wood) are observed at approximately \(45^\circ\) in an L-shaped
corridor. The 60-GHz measurements use Sivers phased-array evaluation kits. The direct TX--RX path is blocked, the test surface is mounted near
the corner, and the 60-GHz receiver is moved over 102 locations in the blocked
branch. The TX is co-located with the Quanergy M8 LiDAR. This campaign is designed to test
whether LiDAR-observed surface characteristics under an oblique geometry are
associated with the surface-mediated 60-GHz response under a controlled NLoS
configuration. Figure~\ref{fig:lcorr} shows the measurement setup.

\begin{figure}[t]
\centering
\InsertCampaignFigure{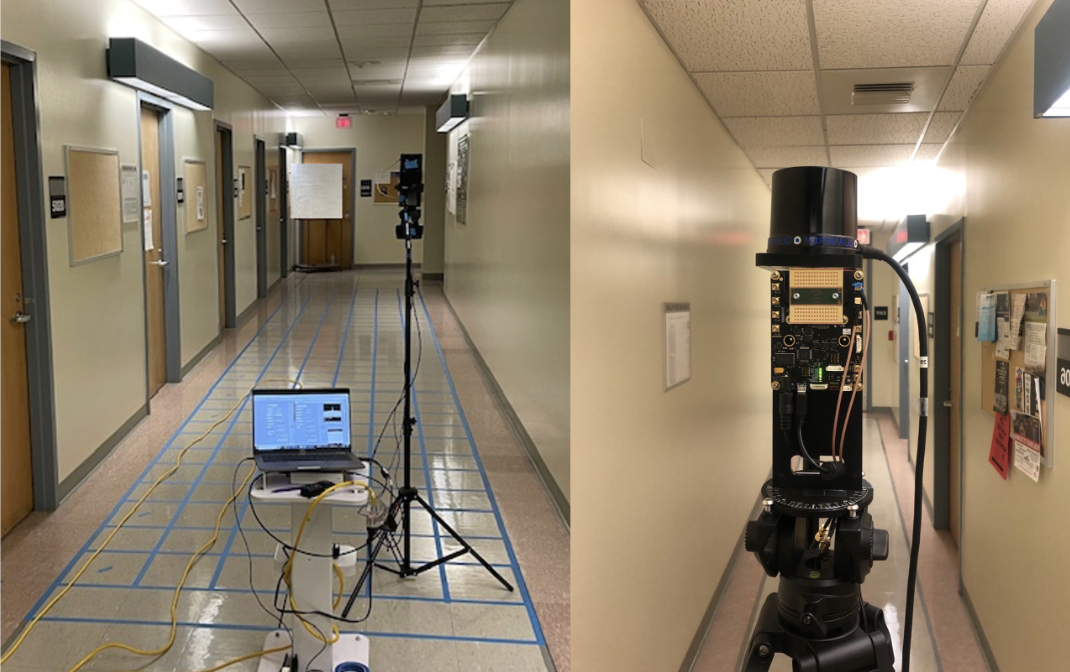}{\columnwidth}
{L-shaped corridor measurement setup.}
\caption{Controlled LiDAR/mmWave corridor campaign. The LiDAR and
60-GHz transmitter are co-located approximately 14~ft~10~in from the test
surface, which is placed near the corner at approximately \(45^\circ\).
The RF receiver samples the resulting surface-mediated field over the
blocked hallway grid.}
\label{fig:lcorr}
\end{figure}

Two full-3-D oblique LiDAR observations are available: a short-range scan at
approximately 3~ft~8~in and a long-range observation corresponding to the
corridor configuration. Together with the near-normal 3~ft~4~in reference,
these observations provide two distinct comparisons. Near-normal versus
short-range oblique primarily tests a viewing-condition change at comparable
range, whereas short- versus long-range oblique tests the effect of a larger
change in acquisition scale while retaining the approximately \(45^\circ\)
surface orientation. Because the sensing conditions and illuminated
footprints differ, these are controlled comparisons rather than a pure
one-factor calibration. The RF experiment keeps the surface position and orientation fixed
while changing the material. All cross-surface RF comparisons use a fixed TX AoD.

\subsection{Independent Room-Scale Campaign}
\label{subsec:room_campaign}

The third campaign uses the same Sivers 60-GHz phased-array kits and a
Quanergy M8 co-located with the RF receiver. The TX remains fixed while the
RX/LiDAR assembly is moved over the room grid in Fig.~\ref{fig:room_env}.
The TX codebook contains 19 sectors spanning
\(-45^\circ:5^\circ:+45^\circ\), and the RX codebook contains 36 sectors
spanning \(-170^\circ:10^\circ:+180^\circ\). A complete
\(19\times36\) RF beam-power map and a full 3-D LiDAR scan are available at
each of the 55 measured RX locations.

\begin{figure}[t]
\centering
\InsertCampaignFigure{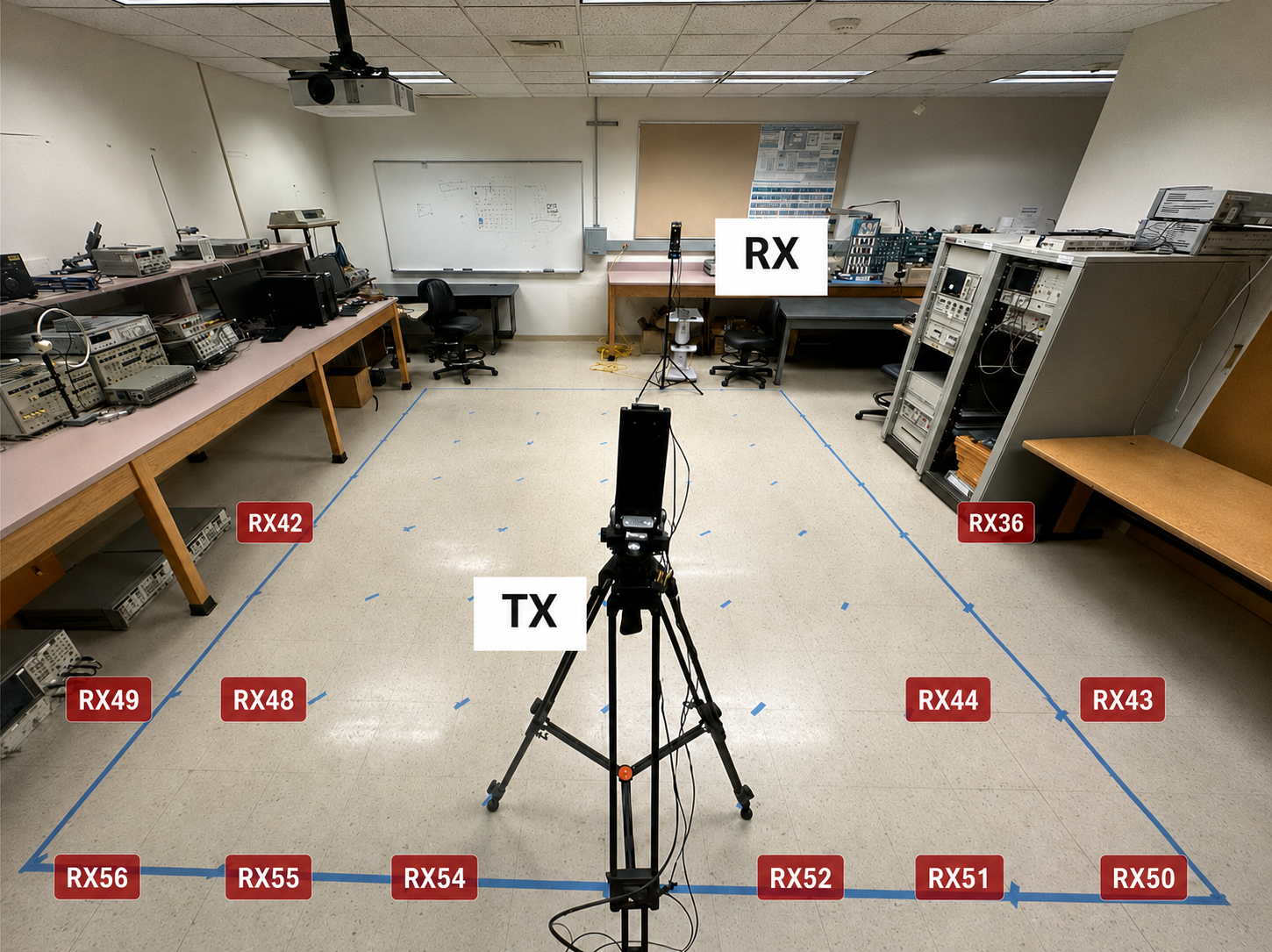}{\columnwidth}
{Room-scale measurement environment.}
\caption{Room-scale LiDAR/mmWave environment. The TX is fixed while the
co-located RX and LiDAR are moved across the measurement grid. Labeled
strict-NLoS locations are the primary blockage subset.}
\label{fig:room_env}
\end{figure}

The strict-NLoS subset contains 12 locations: rx36, rx42--44, rx48--52, and
rx54--56. Forty-one locations are labeled LoS: rx01--35, rx37--41, and rx46.
Locations rx45 and rx47 are mixed.

\subsection{Local 3-D Beam Representation and Evaluation Protocol}
\label{subsec:beam_search_protocol}

To form an NLoS-emphasized search region, the forward-facing RX angles
toward the TX are pruned at the LoS locations, leaving the side/rear
directions for secondary-path evaluation. The resulting receive set is $
\mathcal R_N
=
\{r\in\mathcal R_{\rm all}:|\phi_r|>50^\circ\}.$

The primary 3-D representation uses LiDAR rings 5--7, providing local
vertical support around ring 6 (horizontal slice) while limiting the mixing of unrelated
surfaces. Ring-6-only and all-ring representations are retained for
comparison. To control point-density effects, the sample-count-sensitive descriptors are
computed from a common \(N_0=64\) returns per candidate direction. This
choice retains all strict-NLoS candidates while providing sufficient support
for the local 3-D descriptors. The subsampling is repeated 200 times, and the
median value is retained. The evaluated descriptors include 3-D plane-fit RMS,
\(\mathrm{CV}_I\), \(S_{\max/\mu}\), range-normalized apparent return,
and local 3-D planarity \(\Pi_{3\mathrm D}\). 

The evaluated descriptors include 3-D plane-fit RMS, \(\mathrm{CV}_I\),
\(S_{\max/\mu}\), range-normalized apparent return, and local 3-D planarity
\(\Pi_{3\mathrm D}\). TX-independent selection retains the \(K_R=7\)
highest-ranked RX directions and sweeps all 19 TX beams for each retained
direction. The TX location is not used in the LiDAR ranking.  For held-out evaluation, the ranking cue is selected using all other
locations in the corresponding evaluation subset and is then applied to the
held-out location. This leave-one-location-out procedure is repeated for
each location.

% ==========================================================
% IV. CONTROLLED LIDAR SURFACE-SENSING VALIDATION
% ==========================================================

\section{Controlled LiDAR Surface-Sensing Validation}
\label{sec:controlled_validation}

This section presents the controlled results for RQ1 and RQ2.

\subsection{RQ1: 3-D Surface Description Across View and Range}
\label{subsec:rq1_surface_description}

The near-normal campaign first establishes the broader 16-surface reference
in Table~\ref{tab:normal_surface_comparison}. The measurements exhibit
substantial surface-dependent variation in the LiDAR return statistics.
In particular, $S_{\max/\mu}$ spans from $0.99$~dB for carpet to
$22.23$~dB for whiteboard, while $\mathrm{CV}_{I}$ ranges from $0.02$
to $4.13$. Whiteboard, TV screen, linoleum, silver, metal tin, and copper
exhibit substantially larger localized radiometric contrast than most of
the remaining surfaces. This behavior is qualitatively consistent with the observations reported
in SpaceBeam~\cite{woodford2021spacebeam}. In Table~1 of that work,
the authors characterize infrared specularity using the ratio of maximum
measured intensity to average intensity and observe that surfaces with
higher infrared specularity tend to exhibit lower measured mmWave
reflection loss. The numerical values should not be compared directly,
however, because the sensing hardware, acquisition geometry, sampled
surface regions, and RF measurement conditions differ. More importantly,
the optical and mmWave responses are not physically equivalent, so the
comparison supports only the broader observation that LiDAR radiometric
structure can contain surface-dependent information relevant to RF
propagation.

\begin{table}[t]
\centering
\scriptsize
\caption{Near-normal LiDAR descriptors for 16 indoor surfaces. Surfaces
marked by \(^{*}\) are retained for the matched 3-D oblique experiment.}
\label{tab:normal_surface_comparison}
\setlength{\tabcolsep}{8pt}
\renewcommand{\arraystretch}{1.06}
\begin{tabular}{lrrrr}
\toprule
Surface & \(N_S\) & \(R_{\mathrm{rms},L}\) (mm) &
\(\mathrm{CV}_{I}\) & \(S_{\max/\mu}\) (dB)\\
\midrule
Whiteboard          & 4940 & 10.59 & 4.13 & 22.23\\
TV screen           & 4993 & 13.11 & 3.00 & 16.50\\
Linoleum            & 5189 & 12.43 & 1.65 & 15.02\\
Silver\(^{*}\)      & 5066 & 25.34 & 2.51 & 14.01\\
Metal tin           & 4780 & 14.68 & 2.38 & 12.27\\
Copper\(^{*}\)      & 4896 & 23.44 & 2.49 & 11.28\\
Drywall             & 5032 & 10.01 & 0.45 & 4.83\\
Rough wood\(^{*}\)  & 5048 & 11.57 & 0.47 & 4.07\\
Cardboard\(^{*}\)   & 5010 & 9.96  & 0.36 & 3.79\\
Styrofoam           & 4977 & 10.21 & 0.40 & 3.72\\
Corkboard           & 4965 & 9.95  & 0.37 & 3.71\\
Projector screen    & 4994 & 9.58  & 0.40 & 3.63\\
Smooth wood\(^{*}\) & 5073 & 9.54  & 0.38 & 3.60\\
Fabric pinboard     & 4984 & 9.88  & 0.42 & 3.49\\
Concrete wall       & 5025 & 11.11 & 0.39 & 3.43\\
Carpet              & 3193 & 15.67 & 0.02 & 0.99\\
\bottomrule
\end{tabular}
\end{table}

\begin{figure}[t]
\centering
\subfloat[Copper]{
\includegraphics[width=0.47\columnwidth]
{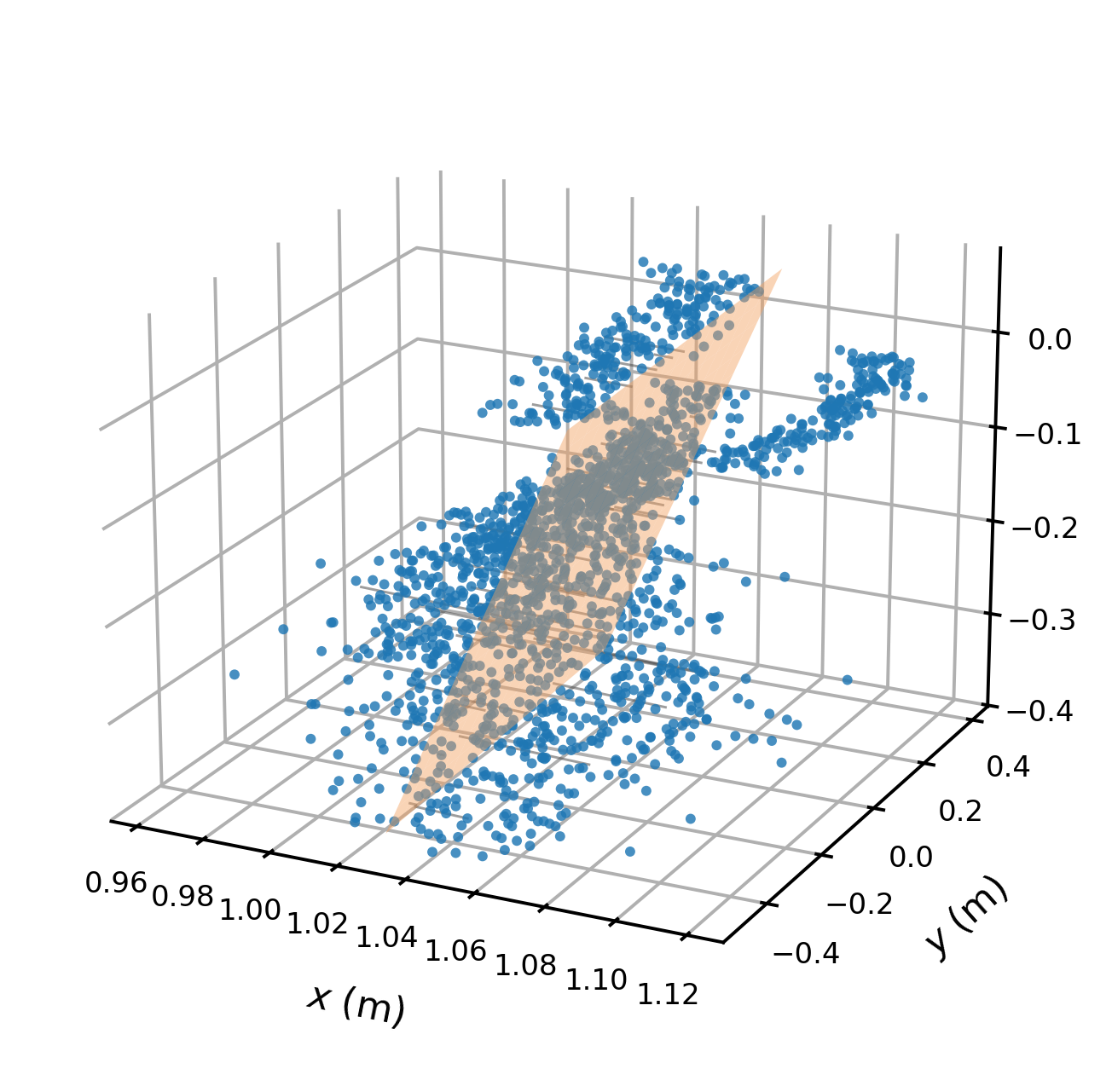}}
\hfill
\subfloat[Silver]{
\includegraphics[width=0.47\columnwidth]
{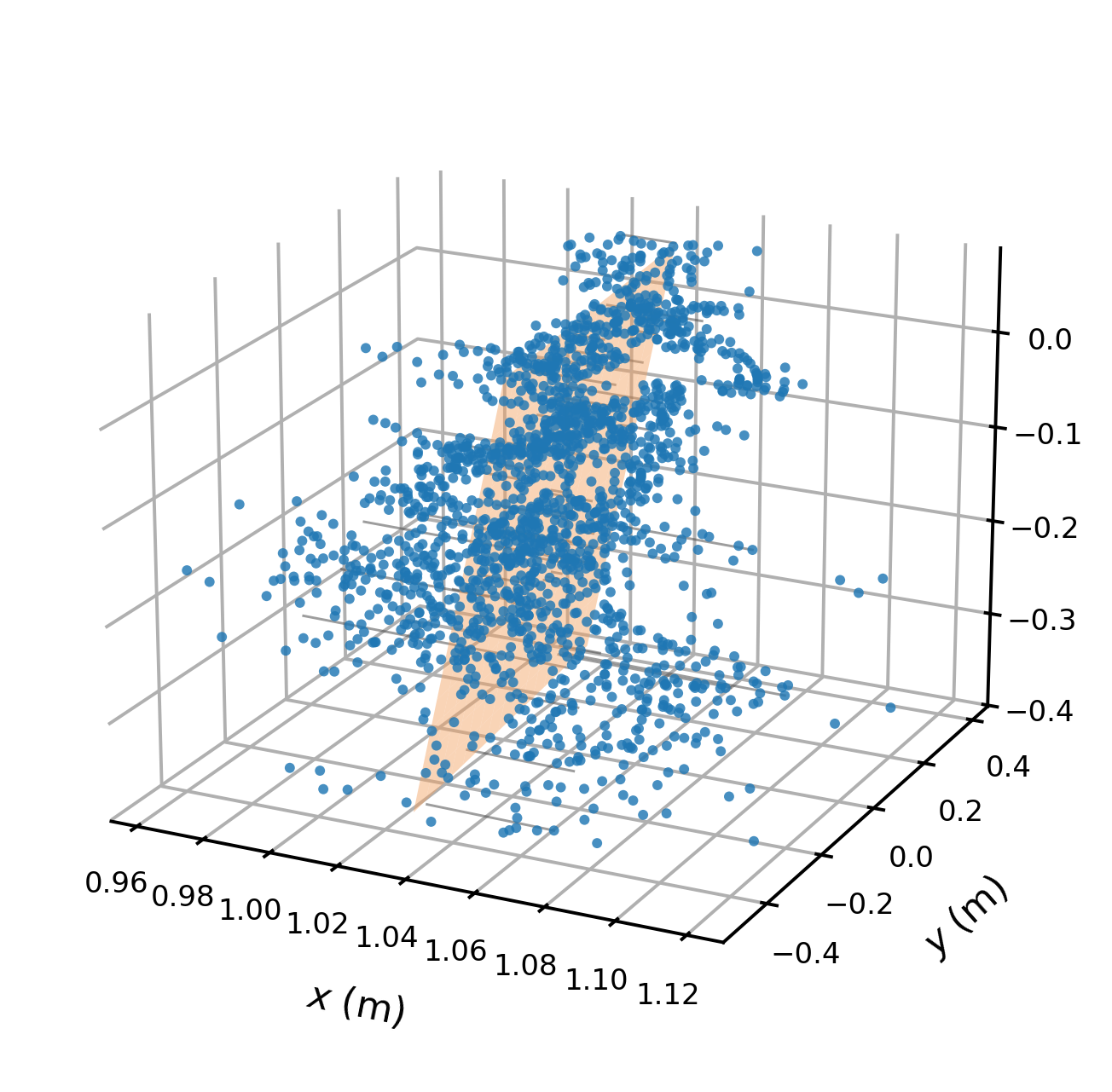}}
\par\vspace{1mm}
\subfloat[Cardboard]{
\includegraphics[width=0.47\columnwidth]
{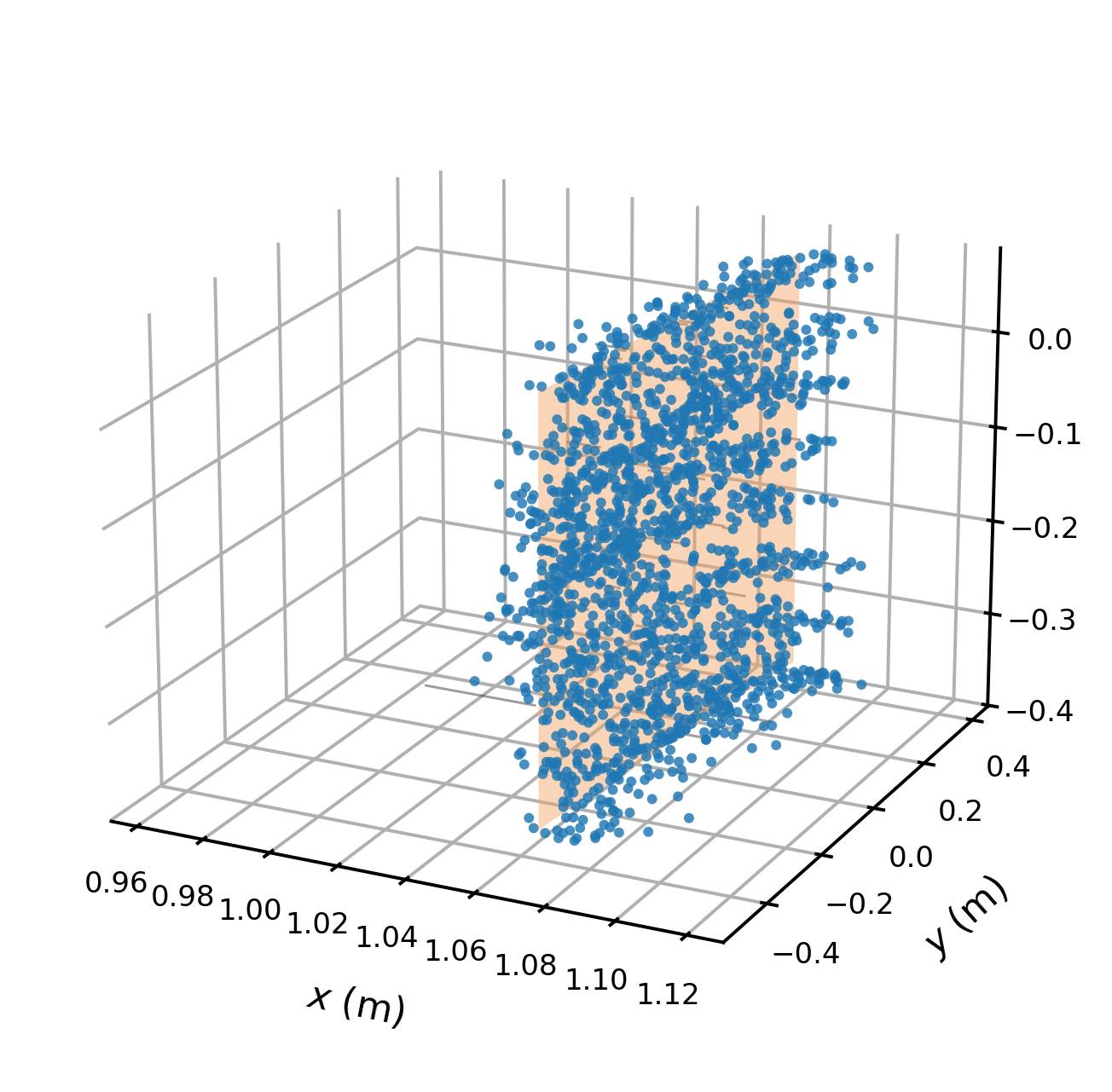}}
\hfill
\subfloat[Rough wood]{
\includegraphics[width=0.47\columnwidth]
{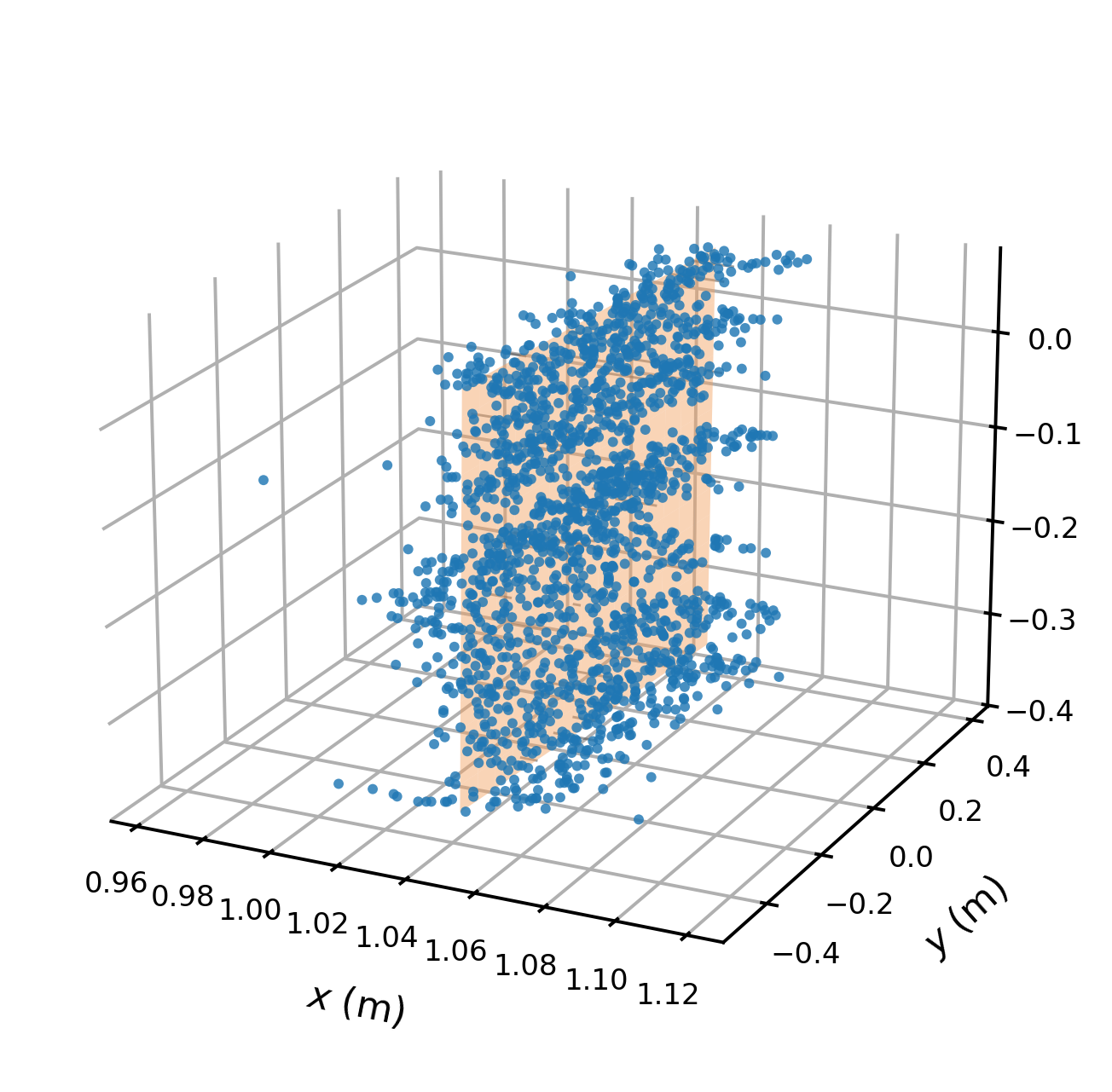}}
\caption{Representative near-normal 3-D LiDAR surface fits. Large residual
dispersion is a property of the measured point cloud and is not interpreted
as RF-scale physical roughness.}
\label{fig:plane_fit_rms_examples}
\end{figure}

\begin{figure*}[t]
\centering
\subfloat[3-D geometric RMS.]{
\includegraphics[width=0.475\textwidth]
{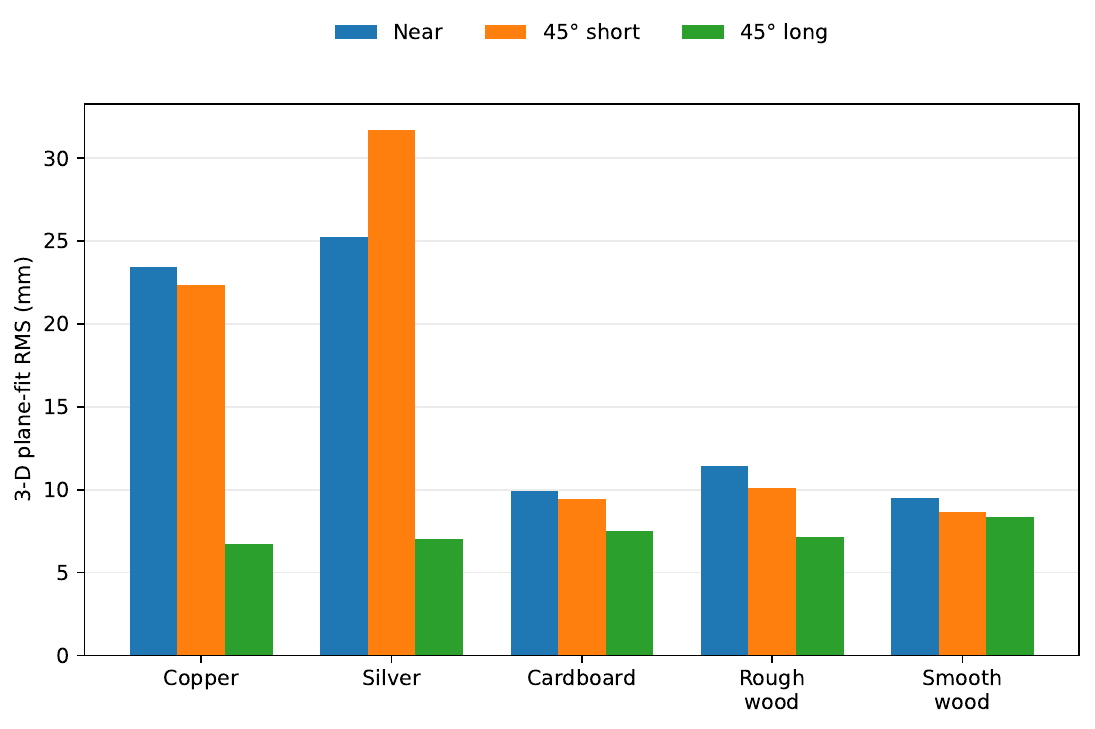}}
\hfill
\subfloat[Max-to-mean radiometric contrast.]{
\includegraphics[width=0.475\textwidth]
{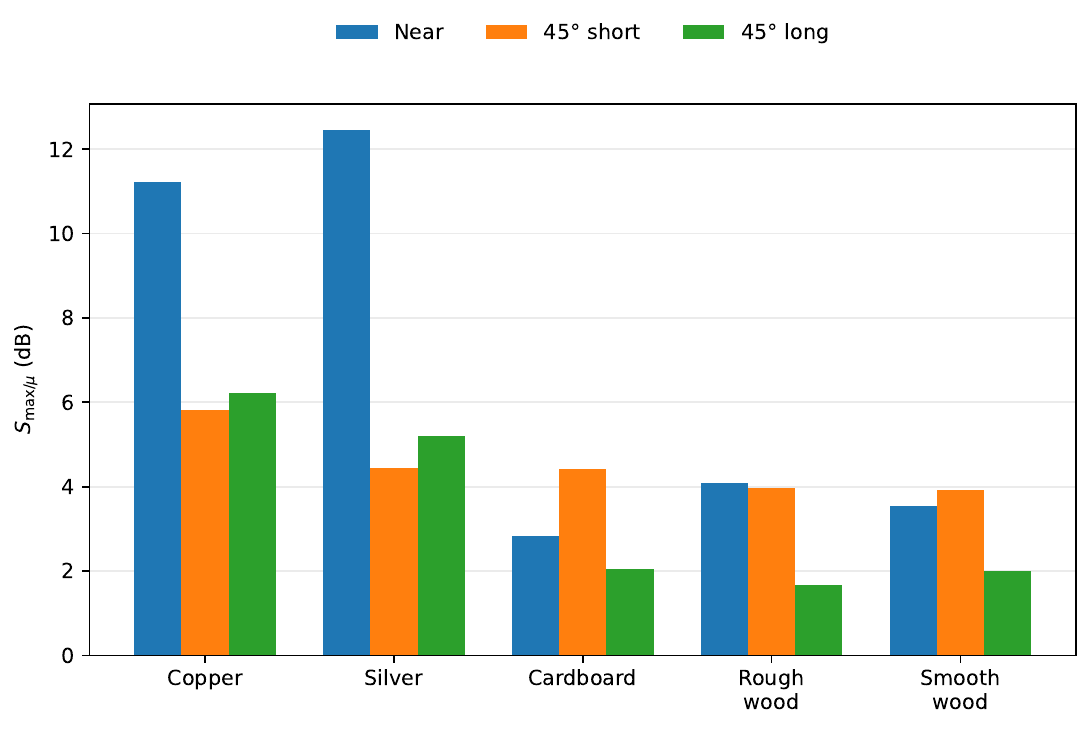}}
\caption{Three-condition full-3-D LiDAR comparison. RMS preserves ordering
across the comparable-range angular change but changes strongly at the long
range; \(S_{\max/\mu}\) retains stronger ordering between the two oblique
ranges.}
\label{fig:three_condition_3d}
\end{figure*}

Figure~\ref{fig:plane_fit_rms_examples} provides representative examples of
the near-normal 3-D plane fits underlying the geometric RMS descriptor.
Copper and silver exhibit larger residual point-cloud dispersion
($23.44$ and $25.34$~mm, respectively) than cardboard and rough wood
($9.96$ and $11.57$~mm, respectively). These differences characterize the
measured LiDAR point-cloud geometry under the given acquisition condition;
they should not be interpreted as differences in microscopic physical
surface roughness at the 60-GHz wavelength scale.  For the five surfaces, Table~\ref{tab:three_condition_3d} compares
the near-normal reference with two full-3-D oblique observations. Throughout
the controlled analysis, \(r\) denotes the Pearson correlation coefficient
and \(\rho\) denotes the Spearman rank correlation coefficient. Because only
five matched surfaces are available, these correlations are interpreted
descriptively.  A three-condition comparison is used rather than a single near/oblique pair
because the long-range observation changes both viewing geometry and
acquisition scale. Geometric RMS exhibits complete rank preservation between
the near-normal and short-range oblique observations (\(\rho=1.00\)), for
which the sensing ranges are comparable. In contrast, the RMS ordering
reverses between the short- and long-range oblique observations
(\(\rho=-0.90\)). This reversal indicates that the long-range change cannot
be attributed to incidence angle alone. Changes in LiDAR footprint, sampling
density, ranging behavior, retained surface support, and sensor processing
with acquisition scale can also affect the measured return-dispersion
descriptor.

\begin{table*}[t]
\centering
\scriptsize
\caption{Full-3-D LiDAR descriptors and cross-condition associations for the
five matched surfaces. The near-normal and short-range oblique observations
were collected at comparable ranges, whereas the third observation represents
the long-range oblique acquisition condition. Correlations are descriptive
because \(n=5\).}
\label{tab:three_condition_3d}
\setlength{\tabcolsep}{3.2pt}
\renewcommand{\arraystretch}{1.05}

\textbf{(a) Full-3-D LiDAR descriptors}\\[2pt]

\begin{tabular}{lrrrrrrrrr}
\toprule
& \multicolumn{3}{c}{Near, 3 ft 4 in}
& \multicolumn{3}{c}{45\(^{\circ}\), 3 ft 8 in}
& \multicolumn{3}{c}{45\(^{\circ}\), long range} \\
\cmidrule(lr){2-4}
\cmidrule(lr){5-7}
\cmidrule(lr){8-10}
Surface
& RMS & CV & \(S_{\max/\mu}\)
& RMS & CV & \(S_{\max/\mu}\)
& RMS & CV & \(S_{\max/\mu}\) \\
& (mm) & & (dB)
& (mm) & & (dB)
& (mm) & & (dB) \\
\midrule
Copper
& 23.46 & 2.493 & 11.23
& 22.36 & 0.519 & 5.81
& 6.70 & 0.747 & 6.22 \\

Silver
& 25.27 & 2.493 & 12.45
& 31.69 & 0.310 & 4.45
& 7.00 & 0.938 & 5.21 \\

Cardboard
& 9.94 & 0.361 & 2.83
& 9.42 & 0.496 & 4.42
& 7.53 & 0.305 & 2.06 \\

Rough wood
& 11.44 & 0.472 & 4.08
& 10.12 & 0.395 & 3.96
& 7.12 & 0.271 & 1.67 \\

Smooth wood
& 9.51 & 0.375 & 3.55
& 8.68 & 0.513 & 3.91
& 8.35 & 0.347 & 2.01 \\
\bottomrule
\end{tabular}

\vspace{5pt}

\textbf{(b) Cross-condition correlation \((r,\rho)\)}\\[2pt]

\begin{tabular}{lccc}
\toprule
Descriptor
& Near--45\(^{\circ}\) short
& Near--45\(^{\circ}\) long
& 45\(^{\circ}\) short--long \\
\midrule
RMS
& \((0.969,\;1.000)\)
& \((-0.738,\;-0.900)\)
& \((-0.643,\;-0.900)\) \\

\(\mathrm{CV}_{I}\)
& \((-0.338,\;-0.103)\)
& \((0.967,\;0.667)\)
& \((-0.452,\;0.000)\) \\

\(S_{\max/\mu}\)
& \((0.650,\;0.600)\)
& \((0.953,\;0.500)\)
& \((0.839,\;0.900)\) \\
\bottomrule
\end{tabular}
\end{table*}

Figure~\ref{fig:three_condition_3d} illustrates the different sensitivities
of the geometric and radiometric descriptors. Geometric RMS preserves the
surface ordering between the near-normal and short-range oblique measurements
($\rho=1.00$), which were collected at comparable ranges, but the ordering
changes substantially at the long-range condition. In contrast,
$S_{\max/\mu}$ changes in magnitude across the measurements but retains a
strong rank association between the two oblique ranges ($\rho=0.90$).
Intensity CV is less consistent across the three conditions. These results
indicate that no single descriptor is robust to all acquisition changes:
geometric RMS is more stable under the tested angular change at comparable
range, whereas max-to-mean radiometric contrast better preserves surface
ranking across the tested oblique range change.

\subsection{RQ2: Controlled 3-D LiDAR/mmWave Association}
\label{subsec:rq2_corridor}

The RF experiment uses the same prescribed corner surface and a fixed TX
AoD. The mmWave transmitter is positioned approximately 14~ft~10~in from
the surface, while the receiver is moved over the 102-point blocked-region
grid in the orthogonal corridor branch. This placement avoids an artificially
close TX--surface configuration and better reflects the operating regime of
interest, in which a useful passive surface may be separated from the
transmitter by several meters. The controlled geometry emphasizes a
surface-mediated TX--surface--RX interaction while allowing the spatial
structure of the resulting 60-GHz field to be measured throughout the
blocked region.

Table~\ref{tab:corridor_rss} summarizes the resulting RF field over the
102-point blocked-region receiver grid. For each surface, coverage is defined
as the percentage of receiver-grid locations at which the measured relative
RSS exceeds \(-70\)~dB. Silver and copper provide the strongest average
responses among the five matched surfaces and each maintains RSS above
\(-70\)~dB over 90.2\% of the blocked-region grid. In contrast, cardboard
and smooth wood produce substantially weaker average fields, with coverage
of 44.1\% and 49.0\%, respectively. For weakly reflecting or partially
transmitting surfaces, the measured field may also contain contributions
from structures behind the test surface and from residual environmental
multipath. The reported RSS should therefore be interpreted as the received
field with the test surface present rather than as an isolated material
reflection coefficient. The substantial spatial variation across the receiver
grid motivates examining RF spatial statistics in addition to mean received
power.

\begin{table}[t]
\centering
\scriptsize
\caption{L-shaped corridor RF response for the five surfaces having matched
3-D LiDAR measurements. Coverage (Cov.) is the percentage of the 102
blocked-region receiver locations with relative RSS greater than \(-70\)~dB.}
\label{tab:corridor_rss}
\setlength{\tabcolsep}{4.2pt}
\renewcommand{\arraystretch}{1.05}
\begin{tabular}{lrrrr}
\toprule
Surface & Mean & Median & Peak & Cov. \(>-70\)\\
& \multicolumn{3}{c}{RSS (dB)} & (\%)\\
\midrule
Silver      & -47.84 & -53.99 & -36.46 & 90.20\\
Copper      & -49.89 & -56.70 & -40.17 & 90.20\\
Rough wood  & -60.13 & -64.48 & -48.12 & 78.43\\
Cardboard   & -66.72 & -71.70 & -57.41 & 44.12\\
Smooth wood & -68.05 & -70.17 & -60.07 & 49.02\\
\bottomrule
\end{tabular}
\end{table}

Table~\ref{tab:cross_modal_3d} compares the long-range full-3-D
LiDAR descriptors with the measured 60-GHz response. Mean RSS characterizes
the average received field over the blocked-region grid, the spatial standard
deviation measures variation across the grid, and the robust
\(P_{90}-P_{10}\) spread reduces the influence of isolated extreme
measurements. Because only five matched surfaces are available, the reported
Pearson and Spearman coefficients are interpreted descriptively rather than
as population-level material relationships.

\begin{table}[t]
\centering
\scriptsize
\caption{Association between the matched long-range full-3-D LiDAR
descriptors and the measured 60-GHz surface-mediated field. RF spatial
standard deviation and \(P_{90}-P_{10}\) are computed over the 102-point
blocked-region receiver grid. }
\label{tab:cross_modal_3d}
\setlength{\tabcolsep}{4.3pt}
\renewcommand{\arraystretch}{1.05}
\begin{tabular}{llrr}
\toprule
LiDAR descriptor & RF metric & \(r\) & \(\rho\)\\
\midrule
\(R_{\mathrm{rms},L}\)
    & Mean RSS
    & -0.838 & -0.900\\
\(R_{\mathrm{rms},L}\)
    & RF spatial std.
    & -0.653 & -0.800\\
\(R_{\mathrm{rms},L}\)
    & RF \(P_{90}-P_{10}\)
    & -0.636 & -0.600\\
\midrule
\(\mathrm{CV}_{I}\)
    & Mean RSS
    & 0.907 & 0.600\\
\(S_{\max/\mu}\)
    & Mean RSS
    & 0.893 & 0.600\\
\(S_{\max/\mu}\)
    & RF \(P_{90}-P_{10}\)
    & 0.865 & 0.900\\
\bottomrule
\end{tabular}
\end{table}
The measurements show that geometric LiDAR return dispersion should
not be interpreted as a scalar reflector-quality metric. Across the five
surfaces, \(R_{\mathrm{rms},L}\) is negatively associated with mean RSS
(\(r=-0.838,\rho=-0.900\)) and with the RF spatial-dispersion measures.
Thus, a larger residual about the fitted LiDAR plane does not imply a
stronger 60-GHz response under this measurement condition.

The radiometric descriptors exhibit a different relationship. Intensity CV
has \(r=0.907\) with mean RSS, while \(S_{\max/\mu}\) has \(r=0.893\) with
mean RSS and a strong rank association with the robust RF spatial spread
(\(\rho=0.900\) for \(P_{90}-P_{10}\)). These observations indicate that
the LiDAR return distribution contains surface-dependent information
associated with the measured 60-GHz field, even though the optical and RF
responses are not physically equivalent.

Figure~\ref{fig:lshape_3d_crossmodal} illustrates these matched cross-modal
relationships. In particular, copper and silver, which produce the strongest
average 60-GHz responses among the five matched surfaces, also exhibit the
two largest \(S_{\max/\mu}\) values. The remaining three surfaces exhibit
smaller max-to-mean contrast and weaker RF responses. Thus,
\(S_{\max/\mu}\) does not predict received-power magnitude or reproduce the
exact RF ordering, but it preserves useful coarse surface ranking for
communication-oriented candidate selection.

\begin{figure*}[t]
\centering
\subfloat[3-D RMS versus mean RSS.]{
\includegraphics[width=0.32\textwidth]
{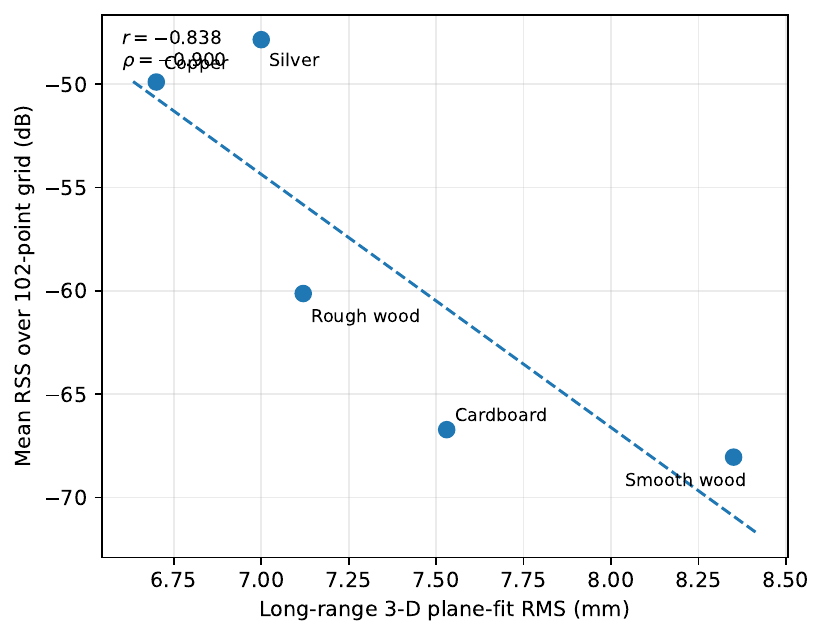}}
\hfill
\subfloat[3-D RMS versus RF spatial standard deviation.]{
\includegraphics[width=0.32\textwidth]
{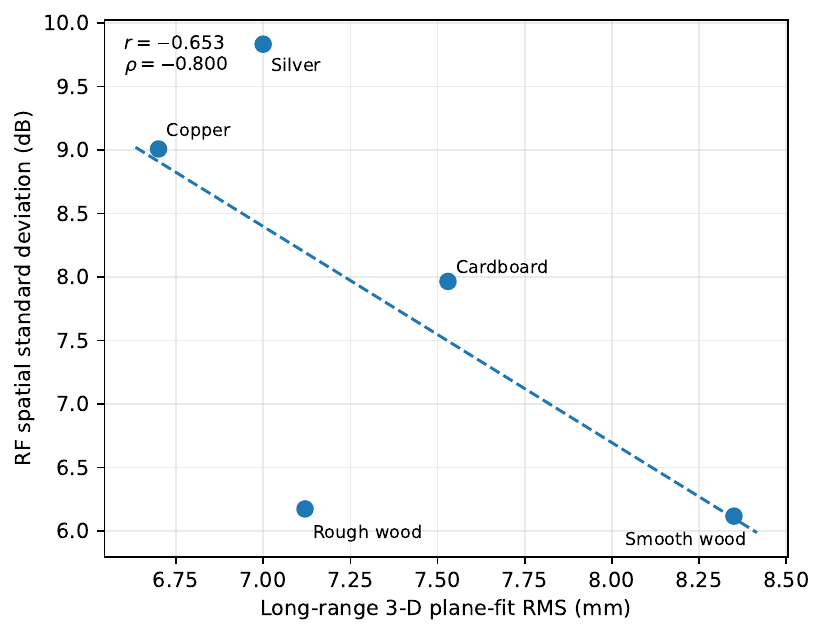}}
\hfill
\subfloat[\(S_{\max/\mu}\) versus robust RF spatial spread.]{
\includegraphics[width=0.32\textwidth]
{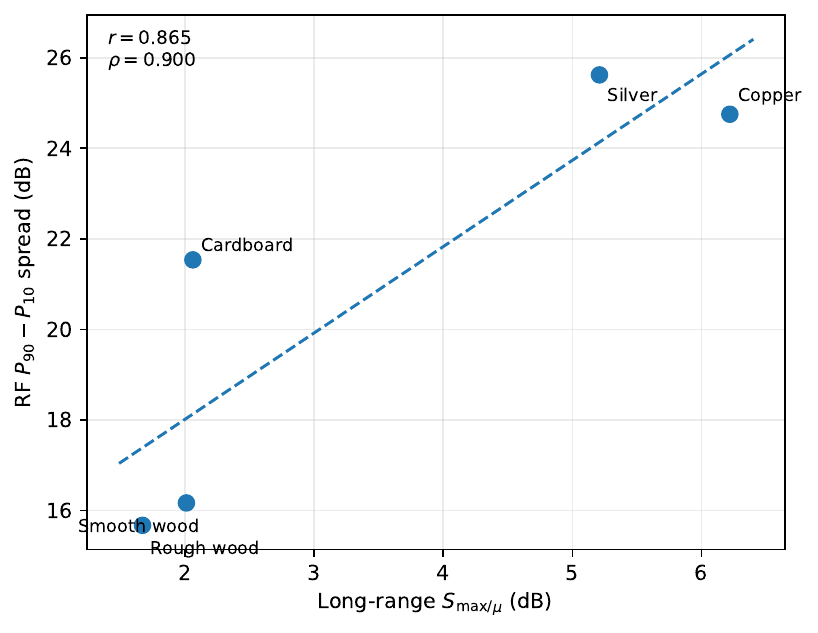}}
\caption{Matched long-range full-3-D LiDAR/mmWave associations across five
controlled surfaces. Geometric RMS does not provide a monotonic reflector
ranking, whereas larger max-to-mean LiDAR contrast is associated with
stronger measured 60-GHz responses. Correlations are descriptive for
\(n=5\).}
\label{fig:lshape_3d_crossmodal}
\end{figure*}

% ==========================================================
% V. ROOM-SCALE VALIDATION
% ==========================================================

\section{Room-Scale Validation: LiDAR-Guided Beam Search}
\label{sec:room_validation}

RQ3 evaluates whether the LiDAR-derived surface prior remains useful in an indoor environment where neither the relevant propagation
surface nor the dominant NLoS path is prescribed. Rather than predicting RF
power directly, the LiDAR observation is used to prioritize receive
directions for subsequent RF probing. Performance is evaluated against the
measured exhaustive \(19\times36\) TX/RX beam map at each room location,
which provides the RF reference for determining whether the LiDAR-guided
search retains a near-optimal beam while reducing the number of directions
that must be probed.

\subsection{LoS-RSS Angular Calibration and Geometry Check}
\label{subsec:room_sanity}

The 41 LoS-labeled locations are first used as a calibration check for the
LiDAR-to-RF angular registration. The direct TX--RX direction is known from
the room geometry and is independently observed through the strongest
measured LoS RSS beam pair. After registration, the median absolute angular
errors are \(2.84^\circ\) at the TX and \(6.34^\circ\) at the RX, and
\(38/41=92.7\%\) of the LoS locations fall within \(\pm10^\circ\) at the TX
and \(\pm20^\circ\) at the RX. No NLoS location is used to fit an additional
angular offset.

With this calibration fixed, we test whether the measured RF beam pair is
compatible with a LiDAR-observed single-bounce surface using tolerances of
\(\pm7.5^\circ\) at the TX and \(\pm15^\circ\) at the RX. Under the full
codebook, \(7/12\) strict-NLoS optima are geometrically compatible. The
AoA-pruned receive-region compatibility is \(40/55=72.7\%\) across all
measured room locations. At LoS locations, the strongest side/rear secondary
component is selected rather than the direct LoS beam. The incompatible cases
may result from unobserved surfaces, multiple interactions, off-specular
scattering, or sidelobe coupling. Thus, the LoS measurements verify the
numerical LiDAR--RF registration, while the compatibility test quantifies how
often the registered visible-surface geometry can explain the measured RF
direction.

\subsection{Three-Ring 3-D Representation and Ablation}
\label{subsec:room_descriptor_protocol}

Table~\ref{tab:representation_fixed} compares fixed LiDAR ranking cues
across the one-ring 2-D, local three-ring 3-D, and unrestricted all-ring
3-D representations at the primary \(K_R=7\) RX-direction probing budget
over the 12 strict-NLoS locations.  The reported percentage denotes the
fraction of locations for which the selected beam is within 3~dB of the
exhaustive-search optimum. The ring-6 2-D coherence \(\Pi(S)\) in
\eqref{eq:lidar_surface_coherence} and the local three-ring 3-D planarity
\(\Pi_{3\mathrm D}(S)\) in \eqref{eq:lidar_3d_planarity} both achieve
\(75.0\%\) within-3-dB recovery with zero median loss. In contrast, fixed
all-ring 3-D planarity achieves \(41.7\%\), while all-ring apparent return
achieves \(75.0\%\). These results show that increasing vertical LiDAR
support does not necessarily improve beam ranking and that its effect
depends on both the descriptor and the point-cloud representation.

\begin{table}[t]
\scriptsize
\caption{Fixed-cue LiDAR representation comparison over the 12 strict-NLoS
locations at \(K_R=7\) and \(N_0=64\). Each entry reports the percentage of
locations with beamforming loss within 3~dB of exhaustive search, with the
median beamforming loss shown in parentheses.}
\label{tab:representation_fixed}
\setlength{\tabcolsep}{5.5pt}
\renewcommand{\arraystretch}{1.08}
\begin{tabular}{lccc}
\toprule
Ranking cue & One-ring 2-D & Three-ring 3-D & All-ring 3-D\\
\midrule

\hspace{-2mm} RMS dispersion
& 41.7\% (4.14~dB)
& 50.0\% (5.71~dB)
& 66.7\% (0.00~dB)\\

 \hspace{-2mm} Intensity CV
& 41.7\% (5.17~dB)
& 33.3\% (5.50~dB)
& 16.7\% (5.76~dB)\\

\hspace{-2mm} Max/mean
& 41.7\% (4.92~dB)
& 16.7\% (8.50~dB)
& 33.3\% (5.76~dB)\\

\hspace{-2mm} Coherence / planarity \hspace{-2mm}
& \textbf{75.0\% (0.00~dB)}
& \textbf{75.0\% (0.00~dB)}
& 41.7\% (6.67~dB)\\

\hspace{-2mm} Apparent return
& 66.7\% (0.67~dB)
& 66.7\% (0.67~dB)
& \textbf{75.0\% (0.00~dB)}\\

\bottomrule
\end{tabular}
\end{table}

Table~\ref{tab:representation_fixed} is therefore used as a fixed-cue
diagnostic to separate the effects of descriptor choice and point-cloud
representation. The following subsection considers a different question; 
whether a ranking cue selected from the remaining locations generalizes to
a held-out location. Equal-budget random probing is evaluated separately in
the system-level comparison of Fig.~\ref{fig:representation_loss_cdf}.

\subsection{Held-Out Three-Ring 3-D Beam Selection}
\label{subsec:heldout_3d}

The objective of this experiment is to determine whether LiDAR beam-level
descriptors can reduce RX beam-search effort at a held-out room location.
The descriptors are used to rank candidate RX directions, and performance is
measured by whether the top \(K_R\) directions contain an RF beam within
3~dB of the exhaustive-search optimum.

Using the leave-one-location-out procedure defined in
Section~\ref{subsec:beam_search_protocol}, the ranking cue for each LiDAR
representation is selected using the remaining locations and then applied
to the held-out location. For the local three-ring representation, planarity
is selected in all 12 strict-NLoS folds. At \(K_R=7\), the resulting search
retains a beam within 3~dB of the exhaustive optimum at
\(9/12=75.0\%\) of the strict-NLoS locations.

The procedure is repeated separately over the 41 LoS-labeled locations
using the side/rear AoA-pruned receive region. Three-ring planarity is again
selected in every fold and achieves \(32/41=78.0\%\) within-3-dB recovery.
A separate held-out evaluation over all 55 measured locations uses the same
AoA-pruned receive region and achieves \(41/55=74.5\%\) recovery with zero
median loss. These are separate within-subset evaluations; the method is not
trained on the strict-NLoS subset and then transfer-tested to the other
locations. Table~\ref{tab:representation_loocv} and
Fig.~\ref{fig:representation_loocv} compare the three LiDAR representations.

\begin{table*}[t]
\centering
\scriptsize
\caption{Held-out TX-independent performance at \(K_R=7\) for the three
LiDAR representations. LoS and All 55 use the common AoA-pruned receive
region.}
\label{tab:representation_loocv}
\begin{tabular}{llrrr}
\toprule
Representation & Selected cue & Strict NLoS & LoS secondary & All 55 pruned\\
\midrule
One-ring 2-D
& Horizontal coherence
& 9/12 (75.0\%)
& 21/41 (51.2\%)
& 29/55 (52.7\%)\\

Three-ring 3-D
& Local planarity
& \textbf{9/12 (75.0\%)}
& \textbf{32/41 (78.0\%)}
& \textbf{41/55 (74.5\%)}\\

All-ring 3-D
& Held-out selected cue
& 9/12 (75.0\%)
& 16/41 (39.0\%)
& 15/55 (27.3\%)\\
\bottomrule
\end{tabular}
\end{table*}

\begin{figure}[t]
\centering
\includegraphics[width=\columnwidth]
{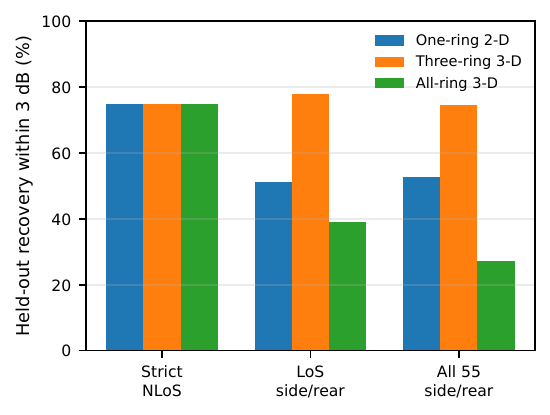}
\caption{Held-out TX-independent beam-selection performance at \(K_R=7\)
for the three LiDAR representations across strict-NLoS, LoS-secondary, and
all-location AoA-pruned evaluations.}
\label{fig:representation_loocv}
\end{figure}

\subsection{Comparison with a Known-TX Geometry Baseline}
\label{subsec:known_tx_surface_selection}

\begin{figure}[t]
\centering
\includegraphics[width=\columnwidth]
{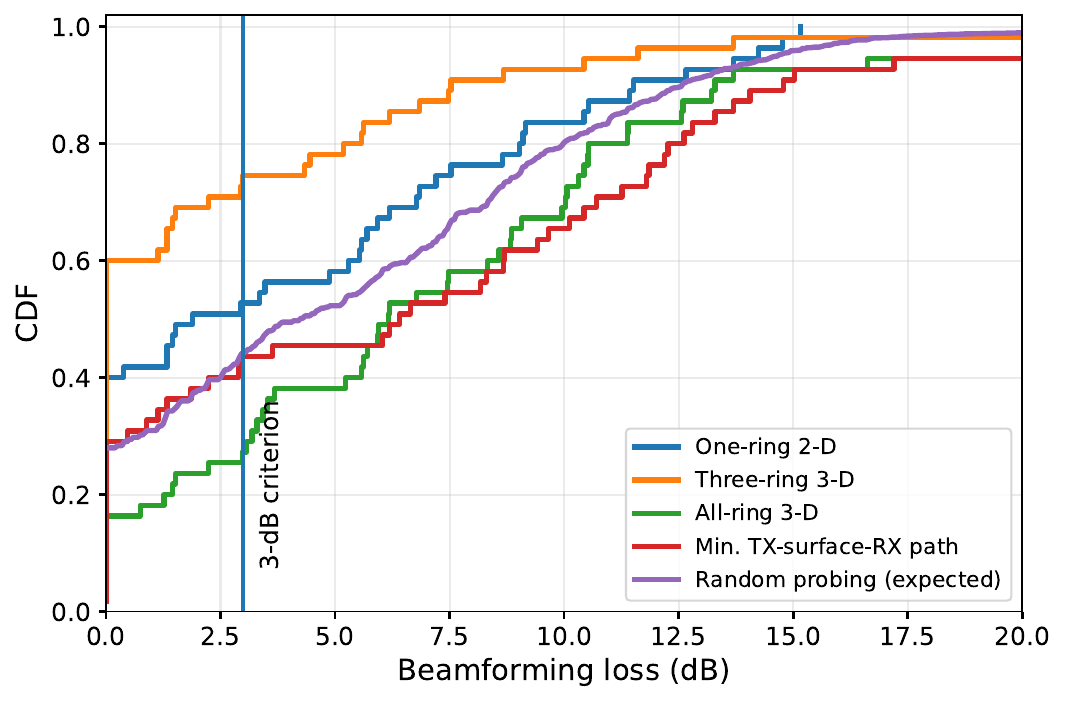}
\caption{CDF of beamforming loss at \(K_R=7\) over all 55 locations using
the common AoA-pruned receive region. The three-ring LiDAR result uses
held-out cue selection; the known-TX minimum-path geometry and random
baselines use the same seven-direction probing budget.}
\label{fig:representation_loss_cdf}
\end{figure}

Figure~\ref{fig:representation_loss_cdf} compares the cumulative distribution
of beamforming loss over all 55 measured locations under the common
\(K_R=7\) RX-direction probing budget. The held-out three-ring 3-D planarity
method achieves within-3-dB recovery at \(41/55=74.5\%\) of the locations,
compared with \(43.6\%\) for the known-TX minimum TX--surface--RX path
baseline and \(44.3\%\) for equal-budget random probing. The geometry
baseline uses additional knowledge of the TX location to rank candidate
paths, whereas the three-ring method ranks RX directions from locally
observed LiDAR surface structure without assuming the upstream propagation
path.  The results therefore show that local 3-D surface structure can provide a
useful beam-search prior even when explicit TX--surface--RX path geometry
is not used.

% ==========================================================
% VI. DISCUSSION
% ==========================================================

\section{Discussion}
\label{sec:discussion}

The three measurement campaigns provide affirmative but bounded answers to
RQ1--RQ3 and demonstrate a multimodal approach for exploiting naturally
occurring passive reflectors in indoor environments. In this framework,
LiDAR provides environmental information about surface geometry and return
structure, while the mmWave radio determines the communication utility of
the observed surfaces through RF probing. The central result is therefore
not that LiDAR predicts 60-GHz received power, but that complementary
information from the two modalities can reduce uncertainty about which
visible surfaces and directions are worth probing. Table~\ref{tab:rq_summary}
summarizes the experimental evidence and the resulting answers to the three
research questions.

\begin{table*}[t]
\centering
\scriptsize
\caption{Summary of the experimental answers to RQ1--RQ3 and their
implications for LiDAR--mmWave beam management.}
\label{tab:rq_summary}
\setlength{\tabcolsep}{6pt}
\renewcommand{\arraystretch}{1.12}
\begin{tabular}{
p{0.055\textwidth}
p{0.48\textwidth}
p{0.39\textwidth}}
\toprule
\textbf{RQ} &
\textbf{Key Experimental Evidence} &
\textbf{Answer and IoT Implication} \\
\midrule

RQ1 &
Descriptor stability depends on the LiDAR acquisition condition. Geometric
RMS preserves surface ordering at comparable range, whereas
\(S_{\max/\mu}\) shows stronger rank persistence across the two tested
oblique ranges. &
\textbf{Descriptor behavior is acquisition dependent.} Surface priors should
therefore be derived from the current LiDAR observation rather than treated
as fixed material properties. \\

\midrule

RQ2 &
In the matched corridor measurements, copper and silver produce the strongest
average 60-GHz fields and also exhibit the largest LiDAR max-to-mean
radiometric contrast, although the exact RF ordering is not preserved. &
\textbf{LiDAR provides surface-ranking information but does not directly
predict the 60-GHz response.} This supports LiDAR-guided candidate
prioritization followed by RF verification. \\

\midrule

RQ3 &
With \(K_R=7\), held-out three-ring 3-D planarity retains a beam within
3~dB of exhaustive search at \(41/55=74.5\%\) of the measured locations
under the common AoA-pruned receive-region evaluation. &
\textbf{Local LiDAR structure can reduce beam-search uncertainty.} LiDAR
prioritizes RX directions for RF probing, while the mmWave measurements
determine the final beam. \\

\bottomrule
\end{tabular}
\end{table*}

\subsection{Acquisition Dependence and Cross-Modal Interpretation}

RQ1 shows that the measured LiDAR descriptors should be interpreted as
properties of the observed surface under the current sensing condition rather
than as intrinsic material constants. In particular, the near-normal to
short-range oblique experiment preserves the geometric RMS ordering, whereas
the long-range oblique condition changes it substantially. The change cannot
be attributed to incidence angle alone because range, footprint, sampling,
retained support, and sensor response also change. Consequently,
\(R_{\mathrm{rms},L}\) represents geometric return dispersion rather than
microscopic RF-scale roughness. The radiometric results provide a more useful communication-oriented
observation. Although \(S_{\max/\mu}\) is not invariant in magnitude, it
retains stronger ordering across the tested oblique conditions and, in RQ2,
places copper and silver surfaces above the three weaker RF surfaces. This
does not establish an optical-to-RF reflection law. Instead, it shows that
one sensing modality can provide useful prior information for another:
LiDAR can rank naturally occurring candidate reflectors, while RF measurements
determine which surface actually supports the useful communication link.

\subsection{Local 3-D Structure for Beam Search}

RQ3 extends this multimodal surface-prior concept to a room where neither the
relevant surface nor the propagation path is prescribed.  The room results show that a local three-ring representation is preferable to
unrestricted vertical aggregation. Combining all eight rings can mix floor,
walls, furniture, and other structures within the same azimuth sector,
whereas the three-ring representation better preserves local surface
organization.

The held-out beam-selection results provide the communication-level evidence.
Using only the top \(K_R=7\) LiDAR-ranked RX directions, three-ring
planarity retains a beam within 3~dB of the exhaustive optimum at
\(75.0\%\) of strict-NLoS locations, \(78.0\%\) of the LoS locations when
searching for side/rear secondary components, and \(74.5\%\) over all
55 measured locations. These results illustrate the complementary roles of
the two modalities.

\subsection{Surface Characterization Rather Than Material Classification}

The proposed approach does not require semantic material identification.
Instead, the objective is communication-oriented surface characterization:
identify naturally available indoor surfaces that are geometrically plausible
and structurally promising enough to merit RF probing. Exact material identity
is neither required nor inferred.

This distinction is important because optical and 60-GHz responses depend on
different wavelengths, material properties, polarizations, and propagation
geometries. A surface that is prominent in LiDAR may not provide the strongest
RF reflection toward a particular receiver. The multimodal framework therefore
does not require a direct mapping between LiDAR return strength and RF power.
LiDAR provides a prior for identifying candidate passive reflectors, while RF
measurements remain the final indicator of communication utility.

\section{Conclusion}
\label{sec:conclusion}

This paper developed and experimentally evaluated a multimodal
LiDAR--mmWave framework for sensing-assisted beam search in indoor
60-GHz networks. The framework proposes exploiting indoor reflective surfaces to support
NLoS beam search.  LiDAR provides geometric and radiometric
surface information to prioritize candidate RF directions, while RF
measurements determine the final communication beam. The controlled experiments show that LiDAR provides useful
surface-dependent information for communication-oriented ranking. Geometric return dispersion remains consistent at comparable LiDAR ranges,
while max-to-mean radiometric contrast is more persistent across the tested
oblique ranges and is associated with stronger measured 60-GHz responses.  In the
room-scale experiment, local three-ring 3-D planarity retains a beam within
3~dB of exhaustive search at \(74.5\%\) of the measured locations under the
common AoA-pruned evaluation using \(K_R=7\) RX directions.

Future work will extend the framework to broader indoor environments and
adaptive 3-D surface representations. Overall, the results support using
LiDAR-observed surface structure to prioritize candidate RF directions,
followed by targeted RF probing for final beam selection.

\balance

\end{document}